\documentclass[11pt]{asaproc}

\usepackage{graphicx}

\usepackage{times}

\usepackage{amsthm,amsmath,amssymb}
\usepackage{graphicx}

\usepackage{subfigure}

\usepackage[authoryear]{natbib}

\usepackage{bbm}

\newtheorem{theorem}{Theorem}

\newtheorem{lemma}{Lemma}
\newtheorem{proposition}{Proposition}

\theoremstyle{remark}
\newtheorem{remark}{Remark}

\theoremstyle{definition}
\newtheorem{definition}{Definition}

\title{Spatial statistics, image analysis and percolation theory}

\author{Mikhail Langovoy\thanks{Max Planck Institute for Intelligent Systems, Department Empirical Inference, and Max Planck Institute for Developmental Biology,  Spemannstrasse 38, D-72076,  T\"{u}bingen, Germany} \and  Michael Habeck\thanks{Max Planck Institute for Intelligent Systems, Department Empirical Inference, and Max Planck Institute for Developmental Biology,  Spemannstrasse 35, D-72076,  T\"{u}bingen, Germany} \and  Bernhard Sch\"{o}lkopf\thanks{Max Planck Institute for Intelligent Systems, Department Empirical Inference,  Spemannstrasse 38, D-72076,  T\"{u}bingen, Germany}
}
\begin{document}

\maketitle

\begin{abstract}
We develop a novel method for detection of signals and reconstruction of images in the presence of random noise. The method uses results from percolation theory. We specifically address the problem of detection of multiple objects of unknown shapes in the case of nonparametric noise. The noise density is unknown and can be heavy-tailed. The objects of interest have unknown varying intensities. No boundary shape constraints are imposed on the objects, only a set of weak bulk conditions is required. We view the object detection problem as a multiple hypothesis testing for discrete statistical inverse problems. We present an algorithm that allows to detect greyscale objects of various shapes in noisy images. We prove results on consistency and algorithmic complexity of our procedures. Applications to cryo-electron microscopy are presented.
\end{abstract}



\section{Statistical model\label{Section_Model}}

Suppose we have a noisy two-dimensional pixelized image. In the present paper we are interested in detection of multiple objects that have an \emph{unknown} colour and unknown, possibly different, shapes. The colour of the objects has to be different from the colour of the background. Thus, we are facing a nonparametric multiple testing problem.

Mathematically, we formalize our problem as follows. We have an $N \times N$ array of observations, i.e. we observe $N^2$ real numbers ${\{Y_{ij}\}}_{i,j=1}^N$. Denote the true value on the pixel $(i, j)$, $1 \leq i, j \leq N$, by $Im_{ij}$, and the corresponding noise by $\varepsilon_{ij}$. Suppose that the noise on the whole screen is independent and identically distributed with an \emph{unknown} distribution function $F$. According to the above,

\begin{equation}\label{1}
Y_{ij} = Im_{ij} +  \varepsilon_{ij}\,,
\end{equation}

\noindent where $1 \leq i, j \leq N$, and $\{ \varepsilon_{ij}\}$, $1 \leq i, j \leq N$ are i.i.d., and $Im_{ij}$ denotes the true value on pixel $(i, j)$. Assume additionally that for the noise

\begin{equation}\label{2}
\varepsilon_{ij} \sim F, \quad \mathbb{E}\, \varepsilon_{ij} = 0,  \quad Var\, \varepsilon_{ij} = \sigma^{2} < + \infty\,.
\end{equation}

Our next assumption is inspired by images from cryo-electron microscopy. We assume that pixels in all particles have color intensity $b$, where $b > a$. In other words,

\begin{equation}\label{3}
\left\{
           \begin{array}{ll}
           Im_{ij} =  b > a, & \hbox{if $(i,j)$ belongs to a particle;} \\
           Im_{ij} = a, & \hbox{if $(i,j)$ does not belong to any particle.}
           \end{array}
         \right.
\end{equation}

\noindent Here we assume that $a$ and $\sigma^2$, as well as the particle intensity $b$, are all unknown. The difficulty in estimating these parameters is that a priori we do not know which pixels contain pure noise and which pixels actually belong to some particle. Moreover, locations, shapes and exact sizes of particles are assumed to be unknown. The number of particles is also unknown, and in this paper we make no probabilistic assumptions about this number or about the distribution of particle locations. Thus we consider the case of a fully nonparametric noise of unknown level. Consistent estimates for $a$, $\sigma$ and $b$ were constructed in \cite{langovoy_habeck_schoelkopf_WSC} and \cite{langovoy_habeck_scan_estimators}. These estimators were called \emph{spatial scan estimators}.

We make now the following assumption. Let there exist a square $K_0$ of size at least $\varphi_0 (N) \times \varphi_0 (N)$ pixels such that $K_0$ doesn't intersect with any particle, and

\begin{equation}
\lim_{N \rightarrow \infty} \varphi_0 (N) \,=\, \infty \,.
\end{equation}

\noindent Here we only require that somewhere, in between the particles, there is one single square of size at least $\varphi_0 (N) \times \varphi_0 (N)$, such that this square is completely filled with noise. There could be several squares of this kind, or there could be bigger squares filled with noise. Locations of noisy squares are unknown. At the first step, our spatial scan estimator is going to find at least one of those noisy patches.

If a pixel $(i, j)$ is white in the original image, denote the corresponding probability distribution of $Y_{ij}$ by $P_0$. For a black pixel $(i, j)$ we denote the corresponding distribution by $P_1$. For any set of pixels $K$, we denote its cardinality by $|K|$. The empirical mean $\overline{Y}_K$ of all the values in $K$ is

\begin{equation}\label{5}
\overline{Y}_K \,=\, \frac{\,1\,}{\,|K|\,} \, \sum_{(i, j) \,\in\, K} Y_{ij} \,.
\end{equation}

\noindent The total sum of values in $K$ is denoted as

\begin{equation}\label{6}
S_K \,=\, \sum_{(i, j) \,\in\, K} Y_{ij} \,.
\end{equation}

Consider now some square $K$ of size $\varphi_0 (N) \times \varphi_0 (N)$ pixels. Suppose that $K$ contains $S_1^{(K)}(N)$ pixels from all of the particles in the image. The remaining $\varphi_0^2 (N) - S_1^{(K)}(N)$ pixels of $K$ contain pure noise. Obviously,

\begin{equation}
S_1^{(K)}(N) \,\leq\, \varphi_0^2 (N)
\end{equation}

\noindent and, by the definition of $K_0$, $ S_1^{(K_0)}(N) \,=\, 0 \,$. In \cite{langovoy_habeck_scan_estimators}, the following auxiliary statement was proved.

\begin{proposition}\label{Proposition_1}

 1) If for some $K$ we have $\lim_{N \rightarrow \infty} \, S_1^{(K)}(N)\,/\,\varphi_0^2 (N) \,=\, 0$, then $\overline{Y}_{K}$ is a consistent estimate of $a$.


\noindent 2) In case if $\lim_{N \rightarrow \infty} \, S_1^{(K)}(N)\,/\,\varphi_0^2 (N) \,\neq\, 0$, it happens that $\overline{Y}_{K}$ is \emph{not} a consistent estimate of $a$.

\end{proposition}

\noindent This proposition clarifies why a naive approach of estimating $a$ by simple averaging of values over the whole screen leads to a generally inconsistent estimate. The naive estimator is consistent only when a combined size of all particles is negligible compared to the size of the screen. On the other hand, we see that sometimes the averaging procedure is consistent, namely, when the averaging is done over some $K$ such that this square has relatively little intersection with particles. This observation was crucial for our construction of scan estimators.

\section{Spatial scan estimators\label{Scan_Estimators}}

\subsection{Spatial scan estimators for the noise}


%
%
%
%
%
%

\begin{definition}\label{Definition_Spatial_Scan_Estimator}

Let $\mathcal{K}_0$ be the collection of all $\varphi_0 (N) \times \varphi_0 (N)$  subsquares of the screen. Set

\begin{equation}\label{Equation_Noisy_Patch}
\widehat{K} \,=\, \arg\min_{K \subseteq \mathcal{K}_0} \{ S_K \} \,.
\end{equation}

\noindent We define a \emph{spatial scan estimator} (for the lower intensity) as

\begin{equation}
\widehat{a} \,:=\, {\varphi_0 (N)}^{-2} \sum_{v \in \widehat{K}} Y_v \,.
\end{equation}

\end{definition}

\noindent The spatial scan estimator for noisy patches can be computed via the following quasi-algorithm. \\

\noindent \textbf{Algorithm 1 (Spatial scan estimator 1).}

\begin{itemize}
\item Step 1. Calculate means over all $\varphi_0 (N) \times \varphi_0 (N)$  subsquares.

\item Step 2. Select a neighborhood $\widehat{K}$ with the smallest mean.

\item Step 3. Define $\widehat{a} \,:=\, {\varphi_0 (N)}^{-2}\sum_{v \in \widehat{K}} Y_v$.

\end{itemize}

\noindent Since the initial noisy image can be naturally viewed as a square lattice graph, where $k^2$-nearest neighbors correspond to a $k \times k$ subsquare on the screen, we see that the spatial scan estimator (based on a popular method of sliding windows) is a special case of the $k$-NN scan estimator defined in \cite{langovoy_habeck_scan_estimators}.


One can also easily modify the spatial scan estimator to estimate the unknown noise distribution $F$. The difficulty here is again that we do not know which pixels contain pure noise. However, this problem is solved by a natural combination of the spatial scan estimator with the empirical distribution function.

\begin{definition}\label{Definition_Spatial_Scan_Estimator_Distribution}

Let $\widehat{K}$ be defined by (\ref{Equation_Noisy_Patch}). A spatial scan estimator of the noise distribution $F$ would be

\begin{equation}\label{Scan_Estimator_Noise_Distribution}
    \widehat{F} (t) \,:=\, \frac{\, 1 \,}{\, |\widehat{K}| \,} \sum_{v \in \widehat{K}} \mathbbm{1}_{ \{ Y_v \leq t \} }  \,.
\end{equation}

\noindent The following variation gives an unbiased consistent estimator for the noise variance $\sigma^2$:

\begin{equation}\label{Scan_Estimator_Noise_Variance}
    \widehat{\sigma}^2 \,:=\, \frac{\, 1 \,}{\, |\widehat{K}| - 1 \,} \sum_{v \in \widehat{K}} {(\, Y_v - \widehat{a} \,)}^2 \,.
\end{equation}

\end{definition}

\noindent These are consistent estimators of $F$ and $\sigma^2$ correspondingly, see \cite{langovoy_habeck_scan_estimators} for more details. Both scan estimators $\widehat{\sigma}^2$ and $\widehat{F} (t)$ can be easily calculated once $\widehat{a}$ is calculated.

\subsection{Computational complexity of scan estimators}

Scan statistics are sometimes criticized for their relatively high computational cost. It is not an issue for the simple form of scanning that is used in our estimators.

\begin{theorem}\label{Scan_Estimators_Computational_Complexity}
Spatial scan estimators $\widehat{a}$ and $\widehat{\sigma}^2$ can be calculated in $O (N^2)$ arithmetic operations, i.e. there are linear complexity algorithms that compute these estimators. Computational complexity does not depend on $\varphi_0 (N)$.
\end{theorem}

\begin{proof}
It is enough to consider only the algorithm for computing $\widehat{a}$, because the other two modifications obviously have the same order of complexity.

We start with Step 1 of the Algorithm, and propose an algorithm to calculate all the means on $\varphi_0 (N) \times \varphi_0 (N)$ subsquares in linear number of operations. Note that in our problem calculating all the means is equivalent to calculating all the sums over the subsquares.

As the first step in calculating the sums, one can calculate all sums over $1 \times \varphi_0 (N)$ columns of pixels. This can be done efficiently by running 1-dimensional sliding window separately by each $1 \times N$ columns of the whole screen. Within each of the $1 \times N$ columns, one uses $\varphi_0 (N)$ operations to calculate the sum in the upper $1 \times \varphi_0 (N)$ column, and then one needs just 2 operations to move one step down. Since there is $N - \varphi_0 (N) + 1$ little columns within the big one, we will calculate the little column sums within each big column in $O (N)$ operations, and all the column sums in the $N$ big columns will be calculated in $O (N^2)$ operations.

Now we switch to $\varphi_0 (N) \times \varphi_0 (N)$ squares. Computing the sum in the upper left square takes now $\varphi_0 (N)$ additions of column sums. Moving from this subsquare one step to the right will take 2 operations (adding and subtracting a column), so the upper row of squares will be done in $\varphi_0 (N) + 2 * (N - \varphi_0 (N)) \,=\, O (N)$ operations.

Notice that we can repeat this procedure in each of the $N - \varphi_0 (N) + 1$ horizontal rows of subsquares. According to the above, computing sums in all the $N - \varphi_0 (N) + 1$ subsquares on the far left takes altogether $(N - \varphi_0 (N) + 1) * \varphi_0 (N)$ additions of columns, and computing all the horizontal shifts takes $O ( (N - \varphi_0 (N) + 1) * N) \,=\, O (N^2)$ operations. At the same time, this completes Step 1 of the Algorithm, which means that there is a realization of Step 1 in $O (N^2)$ arithmetic operations.

To realize Step 2, we can use any of the standard algorithms that search for the smallest element in the array in linear number of operations (see examples in \cite{Aho_Hopcroft_Ullman}). Since there are $(N - \varphi_0 (N) + 1 )^2 \,=\, O (N^2)$ subsquares, the minimum will be found in $O (N^2)$ operations.

Thus, Steps 1 and 2 of the Algorithm can be realized in $O (N^2)$ operations, and the order of complexity doesn't depend on $\varphi_0 (N)$.
\end{proof}

The algorithm suggested in the proof is a theoretical construct and is not expected to be very efficient in practice. There are quicker algorithms for computing sliding window averages using variations of the Fast Fourier Transform.

\subsection{Spatial scan estimators for particles}

The above construction can be inverted to get a spatial scan estimator for particle intensities.

\begin{definition}\label{Definition_Spatial_Scan_Estimator_Particles}

Let $\mathcal{K}_1$ be the collection of all $\varphi_1 (N) \times \varphi_1 (N)$  subsquares of the screen. Set

\begin{equation}
\widehat{K}_1 \,=\, \arg\max_{K \subseteq \mathcal{K}_1} \{ S_K \} \,.
\end{equation}

\noindent We define a \emph{spatial scan estimator} (for the higher intensity) as

\begin{equation}
\widehat{b} \,:=\, {\varphi_1 (N)}^{-2} \sum_{v \in \widehat{K}_1} Y_v \,.
\end{equation}

\end{definition}

\noindent The spatial scan estimator for particle intensities can be computed via the following algorithm. \\

\noindent \textbf{Algorithm 2 (Spatial scan estimator 2).}

\begin{itemize}
\item Step 1. Calculate means over all $\varphi_1 (N) \times \varphi_1 (N)$  subsquares.

\item Step 2. Select a neighborhood $\widehat{K}_1$ with the largest mean.

\item Step 3. Define $\widehat{b} \,:=\, {\varphi_1 (N)}^{-2} \sum_{v \in \widehat{K}_1} Y_v$.

\end{itemize}

\noindent Naturally, this variation of the scan estimator also has linear computational complexity and the order of complexity doesn't depend on the sliding window side $\varphi_1 (N)$. The proof is completely analogous to that of Theorem \ref{Scan_Estimators_Computational_Complexity}.

\section{Consistency for bounded noise\label{Section_Bounded_Noise}}

In order to establish consistency of our spatial scan estimator, and to derive its rates of convergence, we might impose some additional assumptions on our model. Suppose that for all $i$ and $j$

\begin{equation}\label{18}
    |\, \varepsilon_{ij} \,| \,\leq\, M \quad \textit{almost surely}\,.
\end{equation}

\noindent This assumption is realistic in practice. For example, in the cryo-electron microscopy one often has $M \,=\, 1$.

The following bound is useful in quantifying behavior of spatial scan estimators (see \cite{langovoy_habeck_scan_estimators} for a proof).

\begin{proposition}\label{Proposition_7}
Let $K_0$ be any square with side $\varphi_0 (N)$ completely filled with noise, and let $\mathcal{K}$ be any collection of squares, each of size $\varphi_0 (N) \times \varphi_0 (N)$. Define $C_1 \,=\, 3\, (b - a)^2$, $ C_2 \,=\, 12\, \sigma^2 $, $ C_3 \,=\, 4 M \cdot (b-a) $. Then

\begin{equation}\label{28}
    P\, (\,\textit{some}\,\, K \,\in\, \mathcal{K} \,\, \textit{is chosen over} \,\, K_0 \,) \,\leq\, \sum_{K \,\in\, \mathcal{K}} \, \exp \, \biggr( -\, \frac{\, C_1 \, {\bigr[S_1^{(K)}(N)\bigr]}^2 \,}{\, C_2 \,|\,K \setminus K_0\,|\, +\, C_3 \, S_1^{(K)}(N) \,} \biggr) \,.
\end{equation}

\end{proposition}

\noindent We will also use the notation

\begin{equation}\label{29}
\widehat{a}_K \,:=\, \overline{a}_{\widehat{K}} \,.
\end{equation}


\begin{theorem}\label{Theorem_1}
Suppose that

\begin{equation}\label{38}
\lim_{N \rightarrow \infty} \, \frac{\, \varphi_0 (N) \,}{\, \sqrt{\,\log N\,} \,} \,=\, \infty \,.
\end{equation}

\noindent Then the spatial scan estimator (of lower intensities) $\widehat{a} \,=\, \overline{a}_{\widehat{K}}$ is a consistent estimate of $a$. If

\begin{equation}\label{38_1}
\lim_{N \rightarrow \infty} \, \frac{\, \varphi_1 (N) \,}{\, \sqrt{\,\log N\,} \,} \,=\, \infty \,,
\end{equation}


\noindent then the spatial scan estimator (for higher intensities) $\widehat{b}$ is consistent for $b$.
\end{theorem}

What can be said now about the rate of convergence of $\widehat{a}$ to the true $a$? In case if $K_0$ would be known, $\overline{a}_{K_0}$ would in general be an asymptotically efficient estimate of $a$, and also a $\sqrt{\, |\, K_0 \,| \,}$-consistent ($= \,\varphi_0 (N)$-consistent) estimate. One might expect that the rate might be at least $\frac{\, \sqrt{\,\log N\,} \,}{\, \varphi_0 (N) \,}$. The additional $\sqrt{\,\log N\,}$-factor is often called a "price for adaptation" in the literature, and estimates that slow down by a logarithmic factor are called "adaptive" (in our case, adaptation is to the extra difficulty of not knowing whether each particular observation belongs to the noise or not).

However, it can be shown that the spatial scan estimators $\widehat{a}$, $\widehat{F}$ and $\widehat{b}$ actually achieve the corresponding parametric rates of convergence, and that happens uniformly within nonparametric classes of noise distributions. This is a stronger type of adaptivity. See \cite{langovoy_scan_estimators_rates} for more details.

\section{Simultaneous detection of multiple particles\label{Detection_Multiple_Particles}}


The key idea of our particle detection approach is to threshold the noisy picture at a properly chosen level, and then to use percolation theory to do statistics on black and white clusters of the resulting binary image. More details about our approach can be found in \cite{langovoy_davies_wittich}, \cite{langovoy_report_2009-035}, \cite{langovoy_report_Robust_Detection}, \cite{Langovoy_Wittich_Square}, \cite{langovoy_wittich_report_R} and \cite{langovoy_wittich_report_Realistic_Pictures}. See also related recent work \cite{Arias-Castro_etal_Cluster_Detection} and \cite{Arias-Castro_Grimmett}.

We transform the observed noisy image $\{Y_{i,j}\}_{i,j=1}^{N}$ in the following way: for all $1 \leq i, j \leq N$,\par\smallskip

1. $\quad$ If $Y_{ij} \geq \theta (N)$, set $\overline{Y}_{ij} := 1$ (i.e., in the transformed picture the corresponding pixel is coloured black).\smallskip

2. $\quad$ If $Y_{ij} < \theta (N)$, set $\overline{Y}_{ij} := 0$ (i.e., in the transformed picture the corresponding pixel is coloured white).\smallskip

\noindent The above transformation is called \emph{thresholding at the level} $\theta (N)$. The resulting array $\{\overline{Y}_{i,j}\}_{i,j=1}^{N}$ is called a \emph{thresholded picture}.

One can think of pixels from $\{\overline{Y}_{i,j}\}_{i,j=1}^{N}$ as of colored vertices of a sublattice $G_N$ of some planar lattice. We add coloured edges to $G_N$ in the following way. If any two black vertices are neighbours in the underlying lattice, we connect these two vertices with a black edge. If any two white vertices are neighbours, we connect them with a white edge. The choice of underlying lattice is important: different definitions can lead to testing procedures with different properties, see \cite{langovoy_report_2009-035}, \cite{langovoy_davies_wittich} and \cite{langovoy_report_Robust_Detection}. The method becomes especially robust when we work with an $N \times N$ subset of the \emph{triangular} lattice $\mathbb{T}^2$.

Our goal is to choose a threshold $\theta (N)$ such that

\[
P_0 (\,Y_{ij} \geq \theta (N)\,) < p_{c}^{site} \,,
\]

\noindent but

\[
p_{c}^{site} < P_1 (\,Y_{ij} \geq \theta (N)\, ) \,,
\]

\noindent where $p_{c}^{site}$ is the critical probability for site percolation on $\mathbb{T}^2$ (see \cite{Grimmett}). If this is the case, we will observe a so-called \emph{supercritical} percolation of black clusters within particles, and a \emph{subcritical} percolation of black clusters on the background. There will be a high probability of forming large black clusters in place of particles, but there will be only little black clusters in place of noise. The difference between the two regions is the main component in our image analysis method. \\

\noindent \textbf{Algorithm 3 (Detection of multiple objects).}

\begin{itemize}
\item
Step 0. Estimate $a$ and $b$ via the spatial scan estimator. Set $\theta (N) := ( \widehat{a} + \widehat{b} )/2$.

\item Step 1. Perform $\theta(N)-$thresholding of the noisy picture $\{Y_{i,j}\}_{i,j=1}^{N}$.

\item Step 2. Until all black clusters are found, run depth-first search on the graph $G_N$ of the $\theta(N)-$thresholded picture $\{{\overline{Y}}_{i,j}\}_{i,j=1}^{N}$

\item Step 3. When a black cluster of size $\varphi_1 (N) $ was found, report that there is a particle corresponding to this cluster.

\item Step 4. If no black cluster was larger than $\varphi_1 (N) $, report that there were no particles.
\end{itemize}



\noindent Suppose that the image contains $\pi (N)$ particles. Assume

\begin{equation}
\lim_{N \to\infty} \frac{\,\pi (N)\,}{\,\log N\,} = 0\,.
\end{equation}

\noindent Assume also that each particle covers a square of size at least $\varphi_1 (N) \times \varphi_1 (N)$ pixels, and

\begin{equation}\label{7}
\lim_{N \to\infty} \frac{\varphi_1 (N)}{\log N} = \infty \,.
\end{equation}

Let $A$ be a set of pixels corresponding to some particle. Define a corresponding (random) set

\[
thresh (A) \,:=\, \{\overline{Y}_{ij} \,|\, Y_{ij} \in A\} \,.
\]

\noindent For each single particle $A$, let us say that the particle is \emph{detected} by the Algorithm, if $thresh (A)$ contains a significant black cluster. By significant cluster we mean a black cluster of at least $\varphi_1 (N)$ vertices.

For a collection of several particles, the individual significant clusters can possibly merge into bigger clusters. We allow for this in the detection theorem, and leave the task of particle separation for the future pattern recognition work.

\begin{theorem}\label{Theorem1}
Let the noise be symmetric and satisfy the above assumptions. Then \medskip


\noindent 1.  Algorithm 3 finishes its work in $O(N^2)$ steps, i.e. it has linear computational complexity.\medskip


\noindent 2.  If there were $\pi (N)$ particles in the picture, Algorithm 3 detects all of them with probability that tends to 1 as $N \rightarrow \infty$. Moreover, the probability of missing any particles decreases to 0 as $N \rightarrow \infty$ with the rate no slower than

      \begin{equation}\label{43}
        r (N) \,=\, \exp \bigr(\, \pi (N) \ln 2 - C_1  \varphi_1 (N)\,\bigr) \,.
      \end{equation}


\noindent 3.  If there were no particles in the picture, the probability that the Algorithm 3 falsely detects a particle, doesn't exceed $\exp(-C_2(\sigma) \varphi_1 (N))$ for all $N > N(\sigma)$. \medskip

\noindent The constants $C_1, C_2 > 0$ and $N(\sigma)\in\mathbb{N}$ depend only on $\sigma$ and not on particle shapes or exact sizes.

\end{theorem}

\begin{remark}
The rate bound (\ref{43}) given in this theorem is certainly not the actual rate achieved by the algorithm. This bound serves as a preliminary quantification of a trade-off between the power of a multiple testing method on one side, and the number of particles and their sizes on the other side. Algorithm 3 implements a multiple testing procedure, and the bound shows that, under appropriate conditions, this multiple testing procedure is consistent with reasonably bounded probabilities of Type I and Type II errors.
\end{remark}

\begin{proof}
\emph{Part I}. By Theorem \ref{Scan_Estimators_Computational_Complexity}, Step 0 of the Algorithm requires $O(N^2)$ operations. Once the threshold $\widehat{\theta}$ is computed, Steps 1 - 4 can be shown to require $O(N^2)$ operations altogether. The argument here is the same as that in the proof of Theorem 1 in \cite{langovoy_davies_wittich}.

\emph{Part II}. Let $A$ be any set of pixels. Let $\mathfrak{D}_1 (A)$ denote the following random event:

\[
\mathfrak{D}_1 (A) \,=\, \{\, \omega \,|\,\, thresh(A) \,\, \textit{contains black cluster with at least}\,\, \varphi_1 (N) \,\,\textit{vertices} \,\} \,.
\]

\noindent We have the following lemma.

\begin{lemma}\label{Lemma_1}
Let $A$ be a set of pixels corresponding to some single particle. Then

\begin{equation}\label{39}
    P \bigr(\, \mathfrak{D}_1 (A) \,\bigr) \,\geq\, 1 - \exp \bigr(\, - C_1  \varphi_1 (N)\,\bigr) \,,
\end{equation}

\noindent where $C_1 > 0$ is independent of $N$ and the inequality holds for all $N \geq N_0$.

\end{lemma}

\begin{proof} (of Lemma \ref{Lemma_1}). Theorem 5 in \cite{langovoy_davies_wittich} implies that there is a constant $C^{'}_1 > 0$ such that for all $N \geq N^{'}_0$

\begin{equation}\label{40}
P \bigr(\, \mathfrak{D}_1 (A) \,\bigr) \,\geq\, 1 - ( \varphi_1 (N) + 1 )\,e^{-C^{'}_1\, \varphi_1 (N)}\,.
\end{equation}

\noindent Since

\[
1 - (n+1) \,e^{-C^{'}_1\,n} \,>\, 1 - e^{-C_1\,n}\,,
\]

\noindent for some $C_1$: $0 < C_1 < C^{'}_1$ and for all $n \geq N_0$ with some sufficiently large finite $N_0$, the lemma follows.
\end{proof}

\begin{lemma}\label{Lemma_2}
Let $A_1, A_2, \ldots , A_{\pi (N)}$ be sets of pixels corresponding to $\pi (N)$ particles on the screen. Then

\begin{equation}\label{41}
    P \biggr(\, \bigcap_{i = 1}^{\pi (N)} \mathfrak{D}_1 (A_i) \,\biggr) \,\geq\, 1 - \exp \bigr(\, \pi (N) \ln 2 - C_1  \varphi_1 (N)\,\bigr) \frac{\, 1 - \exp \bigr(\, - C_1 \varphi_1 (N) \pi (N)\,\bigr) \,}{\, 1 - \exp \bigr(\, - C_1 \varphi_1 (N)\,\bigr) \,} \,,
\end{equation}

\noindent where $C_1 > 0$ is independent of $N$ and the inequality holds for all $N \geq N_0$.

\end{lemma}

\begin{proof} (of Lemma \ref{Lemma_2}).
We derive the proof from Lemma \ref{Lemma_1}. Indeed, since particles cannot intersect with each other, $A_1, A_2, \ldots , A_{\pi (N)}$ are pairwise disjoint sets of pixels. Independence of noise variables and Lemma \ref{Lemma_1} now lead to

\begin{eqnarray*}
  P \biggr(\, \bigcap_{i = 1}^{\pi (N)} \mathfrak{D}_1 (A_i) \,\biggr) &=& \prod_{i = 1}^{\pi (N)} P \bigr(\, \mathfrak{D}_1 (A_i) \,\bigr)\\
   &\geq& \prod_{i = 1}^{\pi (N)}  \Bigr( 1 - \exp \bigr(\, - C_1  \varphi_1 (N)\,\bigr) \Bigr)\\
    &=& {\Bigr( 1 - \exp \bigr(\, - C_1  \varphi_1 (N)\,\bigr) \Bigr)}^{\pi (N)} \\
    &=& \sum_{l = 0}^{\pi (N)} {(-1)}^l C^{l}_{\pi (N)} \exp \bigr(\, - C_1 \cdot l \cdot \varphi_1 (N)\,\bigr) \\
    &\geq& 1 - \sum_{l = 1}^{\pi (N)} C^{l}_{\pi (N)} \exp \bigr(\, - C_1 \cdot l \cdot \varphi_1 (N)\,\bigr) \,.
\end{eqnarray*}

\noindent In the above inequalities, $C^{k}_{n}$ denotes the standard binomial coefficient

\[
C^{k}_{n} \,=\, \frac{\, n! \,}{\, k! (n-k)! \,} \,.
\]

\noindent It is easy to show that, for all $1 \leq k \leq n$,

\[
C^{k}_{n} \,\leq\, 2^n \,.
\]

\noindent Plugging this in the above estimates, we get

\begin{eqnarray*}
  P \biggr(\, \bigcap_{i = 1}^{\pi (N)} \mathfrak{D}_1 (A_i) \,\biggr) &\geq& 1 - \sum_{l = 1}^{\pi (N)} C^{l}_{\pi (N)} \exp \bigr(\, - C_1 \cdot l \cdot \varphi_1 (N)\,\bigr) \\
    &\geq& 1 - \sum_{l = 1}^{\pi (N)} 2^{\pi (N)} \exp \bigr(\, - C_1 \cdot l \cdot \varphi_1 (N)\,\bigr) \\
    &=& 1 - \sum_{l = 1}^{\pi (N)} \exp \bigr(\, \pi (N) \ln 2 - C_1 \cdot l \cdot \varphi_1 (N)\,\bigr)  \\
    &=& 1 - \exp (\, \pi (N) \ln 2 \,) \cdot  \exp \bigr(\, - C_1 \varphi_1 (N)\,\bigr) \times \\
    & & \times \frac{\, 1 - \exp \bigr(\, - C_1 \varphi_1 (N) \pi (N)\,\bigr) \,}{\, 1 - \exp \bigr(\, - C_1 \varphi_1 (N)\,\bigr) \,} \,.
\end{eqnarray*}

\noindent The last closed form expression for the exponential sum was obtained from the following geometric series expression:

\[
\sum_{l =1}^{m} \exp (- l x) \,=\, e^{-x} \, \frac{\, 1 - e^{- m x} \,}{\, 1 - e^{- x} \,} \,.
\]

\noindent This finishes the proof of Lemma \ref{Lemma_2}.

\end{proof}

\noindent By a simple passage to the limit, Lemma \ref{Lemma_2} gives us the following rate of convergence for the probability of missing any particle.

\begin{lemma}\label{Lemma_3}
Let $A_1, A_2, \ldots , A_{\pi (N)}$ be sets of pixels corresponding to $\pi (N)$ particles on the screen. Then there exists a constant $C_1 > 0$ such that

\begin{equation}\label{42}
    \limsup_{N \rightarrow \infty} \, \frac{1 - P \biggr(\, \bigcap_{i = 1}^{\pi (N)} \mathfrak{D}_1 (A_i) \,\biggr)}{\,\exp \bigr(\, \pi (N) \ln 2 - C_1  \varphi_1 (N)\,\bigr)\,} \,\leq\,1 \,.
\end{equation}



\end{lemma}

\noindent This completes the proof of Part II of Theorem \ref{Theorem1}, because the event $\bigcap_{i = 1}^{\pi (N)} \mathfrak{D}_1 (A_i)$ is exactly the event that each single particle was detected.

\emph{Part III}. Suppose there were no particles on the screen. In this case, the combined probability that one or more particles were falsely detected, doesn't exceed the probability that $\{ \overline{Y}_{i, j = 1}^N \}$ contains a black cluster of size at least $\varphi_1 (N)$. This last probability has an explicit upper bound that is given, for example, in Theorem 1 of \cite{langovoy_davies_wittich}. This completes the proof of Theorem \ref{Theorem1}.

\end{proof}

%
%
%

\section{Application to particle picking in cryo-electron microscopy}
Cryo-electron microscopy (cryo-EM) is an emerging experimental method to characterize the structure of large biomolecular assemblies.
Single particle cryo-EM records 2D images (so-called micrographs) of projections of the three-dimensional particle, which need to be processed to obtain the three-dimensional reconstruction.
A crucial step in the reconstruction process is particle picking which involves detection of particles in noisy 2D micrographs with low signal-to-noise ratios of typically 1:10 or even lower. Typically, each picture contains a large number of particles, and particles have unknown irregular and nonconvex shapes.

We applied our method to micrographs measured in cryo-electron microscopy.
The results are shown in Figure \ref{fig:cryoem}.
A micrograph of GroEL was downloaded from a public repository for particle picking (http://ami.scripps.edu/redmine/projects/ami/wiki/GroEL\_dataset\_I).
Before running the scan estimators the original $2400 \times 2400$ was down-sampled twice to improve the signal to noise and normalized to a maximum intensity of 1.
Our scan estimators yield $a=0.319$ (using a window size of 65) and $b=0.453$ (using a window with side 9), resulting in the threshold $\theta=0.386$.
After running the percolation analysis on the thresholded image (shown in Figure \ref{subfig:binary}) we delete black clusters occupying less than 30 pixels.
The result of the filtering procedure is shown in Figure \ref{subfig:filtered}. For each of the remaining black clusters, the algorithm claims that there is a particle containing this cluster.
%
%
\begin{figure}
\centering
\subfigure[]{
\includegraphics[scale=0.22]{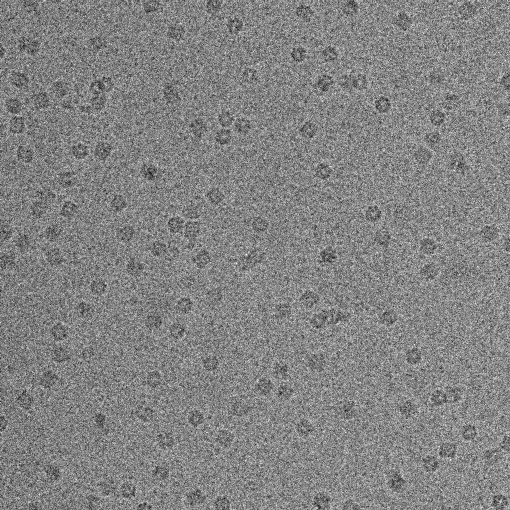}
\label{subfig:original}
}
\subfigure[]{
\includegraphics[scale=0.22]{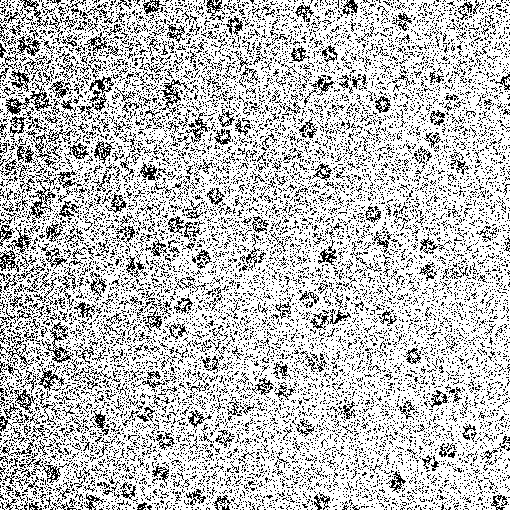}
\label{subfig:binary}	
}
\subfigure[]{
\includegraphics[scale=0.22]{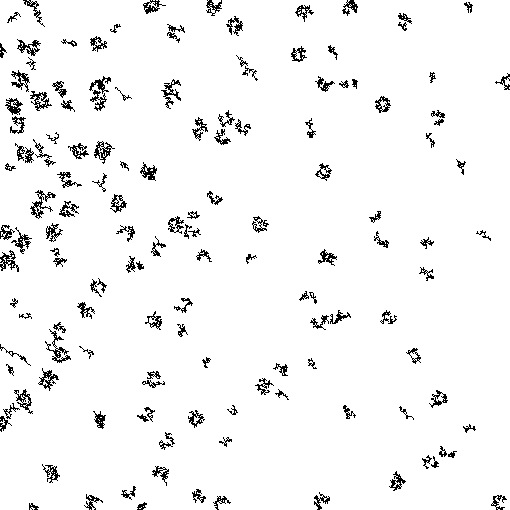}
\label{subfig:filtered}	
}
\caption{%
Application to particle picking in cryo-EM.
(a): Original micrograph with particles of GroEL.
(b): Thresholded image based on spatial scan estimators of the noise and the signal.
(c): Thresholded image with random black clusters removed.
\label{fig:cryoem}}
\end{figure}
%
%

\bibliographystyle{plainnat}
\bibliography{papiere}

\begin{thebibliography}{13}
\providecommand{\natexlab}[1]{#1}
\providecommand{\url}[1]{\texttt{#1}}
\expandafter\ifx\csname urlstyle\endcsname\relax
  \providecommand{\doi}[1]{doi: #1}\else
  \providecommand{\doi}{doi: \begingroup \urlstyle{rm}\Url}\fi

\bibitem[Aho et~al.(1975)Aho, Hopcroft, and Ullman]{Aho_Hopcroft_Ullman}
Alfred~V. Aho, John~E. Hopcroft, and Jeffrey~D. Ullman.
\newblock \emph{The design and analysis of computer algorithms}.
\newblock Addison-Wesley Publishing Co., Reading, Mass.-London-Amsterdam, 1975.
\newblock Second printing, Addison-Wesley Series in Computer Science and
  Information Processing.

\bibitem[Arias-Castro and Grimmett(2011)]{Arias-Castro_Grimmett}
E.~Arias-Castro and G.~Grimmett.
\newblock {Cluster Detection in Networks using Percolation}.
\newblock \emph{Submitted}, 2011.
\newblock URL \url{http://arxiv.org/abs/1104.0338}.

\bibitem[Arias-Castro et~al.(2011)Arias-Castro, Candes, and
  Durand]{Arias-Castro_etal_Cluster_Detection}
E.~Arias-Castro, E.~Candes, and A.~Durand.
\newblock {Detection of an anomalous cluster in a network}.
\newblock \emph{Annals of Statistics}, 39\penalty0 (1):\penalty0 278--304,
  2011.

\bibitem[Davies et~al.(2009)Davies, Langovoy, and
  Wittich]{langovoy_davies_wittich}
P.~L. Davies, M.~Langovoy, and O.~Wittich.
\newblock Randomized algorithms for statistical image analysis based on
  percolation theory.
\newblock \emph{Submitted}, 2009.

\bibitem[Grimmett(1999)]{Grimmett}
Geoffrey Grimmett.
\newblock \emph{Percolation}, volume 321 of \emph{Grundlehren der
  Mathematischen Wissenschaften [Fundamental Principles of Mathematical
  Sciences]}.
\newblock Springer-Verlag, Berlin, second edition, 1999.
\newblock ISBN 3-540-64902-6.

\bibitem[Langovoy(2011)]{langovoy_scan_estimators_rates}
M.~Langovoy.
\newblock Adaptive rate optimality of spatial scan estimators and k-nn scan
  estimators.
\newblock \emph{In preparation}, 2011.

\bibitem[Langovoy and Habeck(2011)]{langovoy_habeck_scan_estimators}
M.~Langovoy and M.~Habeck.
\newblock \emph{Uniform consistency of spatial scan estimators and k-NN scan
  estimators}.
\newblock Technical Report. Max Planck Institute for Intelligent Systems,
  T\"{u}bingen, Germany, 2011.

\bibitem[Langovoy and Wittich(2009{\natexlab{a}})]{langovoy_report_2009-035}
M.~Langovoy and O.~Wittich.
\newblock \emph{Detection of objects in noisy images and site percolation on
  square lattices}.
\newblock EURANDOM Report No. 2009-035. EURANDOM, Eindhoven,
  2009{\natexlab{a}}.

\bibitem[Langovoy and
  Wittich(2010{\natexlab{a}})]{langovoy_report_Robust_Detection}
M.~Langovoy and O.~Wittich.
\newblock \emph{Robust nonparametric detection of objects in noisy images}.
\newblock EURANDOM Report No. 2010-049. EURANDOM, Eindhoven,
  2010{\natexlab{a}}.

\bibitem[Langovoy and Wittich(2010{\natexlab{b}})]{langovoy_wittich_report_R}
M.~Langovoy and O.~Wittich.
\newblock \emph{Computationally efficient algorithms for statistical image
  processing. Implementation in R}.
\newblock EURANDOM Report No. 2010-053. EURANDOM, Eindhoven,
  2010{\natexlab{b}}.

\bibitem[Langovoy and
  Wittich(2011)]{langovoy_wittich_report_Realistic_Pictures}
M.~Langovoy and O.~Wittich.
\newblock \emph{Multiple testing, uncertainty and realistic pictures}.
\newblock EURANDOM Report No. 2011-004. EURANDOM, Eindhoven, 2011.
\newblock URL \url{http://arxiv.org/abs/1102.4820}.

\bibitem[Langovoy et~al.(2011)Langovoy, Habeck, and
  Sch\"{o}lkopf]{langovoy_habeck_schoelkopf_WSC}
M.~Langovoy, M.~Habeck, and B.~Sch\"{o}lkopf.
\newblock Adaptive nonparametric detection in cryo-electron microscopy.
\newblock \emph{To appear in the Bulletin of the Insternational Statistical
  Institute}, 2011.

\bibitem[Langovoy and Wittich(2009{\natexlab{b}})]{Langovoy_Wittich_Square}
M.~A. Langovoy and O.~Wittich.
\newblock Detection of objects in noisy images and site percolation on square
  lattices.
\newblock \emph{Submitted}, 2009{\natexlab{b}}.

\end{thebibliography}

\end{document}